\documentclass{article}
\usepackage{fullpage}

\usepackage{graphicx}
\usepackage{algorithm}
\usepackage{algorithmicx}
\usepackage[noend]{algpseudocode}
\usepackage{amsmath,amssymb,amsthm,mathtools}
\usepackage{paralist}
\usepackage{bm}
\usepackage{xspace}
\usepackage{url}
\usepackage{fullpage, prettyref}
\usepackage{boxedminipage}
\usepackage{wrapfig}
\usepackage{ifthen}
\usepackage{color}
\usepackage{xcolor}
\usepackage{framed}
\usepackage[pagebackref,colorlinks=true,pdfpagemode=none,urlcolor=blue,linkcolor=blue,citecolor=violet,pdfstartview=FitH]{hyperref}
\usepackage{fullpage}
\usepackage{fullpage}

\usepackage{thmtools}
\usepackage{thm-restate}

\algnewcommand{\LineComment}[1]{\State \(\triangleright\) \emph{\color{violet} \small #1}}
\newtheorem{theorem}{Theorem}[section]

\newtheorem{lemma}[theorem]{Lemma}
\newtheorem{claim}[theorem]{Claim}

\newtheorem{definition}[theorem]{Definition}

\newcommand{\ignore}[1]{}

\newcommand{\eps}{\varepsilon}

\newcommand{\poly}{\mathrm{poly}}

\newcommand{\floor}[1]{\lfloor#1\rfloor}

\newcommand{\EX}{\hbox{\bf E}}
\newcommand{\Var}{\hbox{\bf Var}}

\newcommand{\otilde}{\widetilde{O}}

\newcommand{\eqdef}{:=}

\newcommand{\ers}{{\tt ERS}}
\newcommand{\ersg}{{\tt ERS-gen}}

\newcommand{\Sec}[1]{\hyperref[sec:#1]{\S\ref*{sec:#1}}} 
\newcommand{\Eqn}[1]{\hyperref[eq:#1]{(\ref*{eq:#1})}} 
\newcommand{\Fig}[1]{\hyperref[fig:#1]{Fig.\,\ref*{fig:#1}}} 
\newcommand{\Tab}[1]{\hyperref[tab:#1]{Tab.\,\ref*{tab:#1}}} 
\newcommand{\Thm}[1]{\hyperref[thm:#1]{Theorem\,\ref*{thm:#1}}} 
\newcommand{\Fact}[1]{\hyperref[fact:#1]{Fact\,\ref*{fact:#1}}} 
\newcommand{\Lem}[1]{\hyperref[lem:#1]{Lemma\,\ref*{lem:#1}}} 
\newcommand{\Prop}[1]{\hyperref[prop:#1]{Prop.~\ref*{prop:#1}}} 
\newcommand{\Cor}[1]{\hyperref[cor:#1]{Corollary~\ref*{cor:#1}}} 
\newcommand{\Conj}[1]{\hyperref[conj:#1]{Conjecture~\ref*{conj:#1}}} 
\newcommand{\Def}[1]{\hyperref[def:#1]{Definition~\ref*{def:#1}}} 
\newcommand{\Alg}[1]{\hyperref[alg:#1]{Alg.~\ref*{alg:#1}}} 
\newcommand{\Ex}[1]{\hyperref[ex:#1]{Ex.~\ref*{ex:#1}}} 
\newcommand{\Clm}[1]{\hyperref[clm:#1]{Claim~\ref*{clm:#1}}} 
\newcommand{\Step}[1]{\hyperref[step:#1]{Step~\ref*{step:#1}}} 
\newcommand{\Obs}[1]{\hyperref[obs:#1]{Obs~\ref*{obs:#1}}} 

\begin{document}

\title{A note on approximating the average degree of bounded arboricity graphs}
\date{}

\author{Talya Eden\\
Bar-Ilan University\\
{\tt talyaa01@gmail.com}
\and
C. Seshadhri \\
University of California, Santa Cruz\\
{\tt sesh@ucsc.edu}
}

\maketitle

\begin{abstract}
Estimating the average degree of graph is a classic problem
in sublinear graph algorithm. 
Eden, Ron, and Seshadhri (ICALP 2017, SIDMA 2019) gave a simple algorithm
for this problem whose running time depended on the graph arboricity, 
but the underlying simplicity and associated analysis were buried inside the main result.
Moreover, the description there loses logarithmic factors because of parameter search.
The aim of this note is to give a full presentation of this algorithm, without these losses.

Consider standard access (vertex samples, degree queries, and neighbor queries) to a graph
$G = (V,E)$ of arboricity at most $\alpha$. Let $d$ denote the average degree of $G$.
We describe an algorithm that gives a $(1+\eps)$-approximation to $d$
degree using $O(\eps^{-2}\alpha/d)$ queries. For completeness, we modify the algorithm to get a $O(\eps^{-2} \sqrt{n/d})$ query
algorithm for general graphs. 
\end{abstract}

\section{Introduction} \label{sec:intro}

Given a simple graph $G = (V,E)$ with $n$ vertices and $m$ edges,
the average degree is denoted $d \eqdef 2m/n$. The problem of estimating
$d$ in sublinear time goes back to the early days of sublinear graph algorithms.
Feige initiated study on this problem~\cite{Feig06}, and Goldreich-Ron~\cite{GR08} (henceforth GR)
obtained the first sublinear $(1+\eps)$-approximation algorithm for the average degree.
The query complexity (and running time) were $\poly(\eps^{-1}\log n) \sqrt{n/d}$,
and GR gave a nearly matching $\Omega(\sqrt{n/d})$ lower bound. The GR algorithm
is relatively complex, and loses various logarithm and $\eps$ factors due to bucketing
techniques.

Eden-Ron-Seshadhri~\cite{ERS19} (henceforth ERS) gave a significantly simpler algorithm that removes
the bucketing. This algorithm has an easier analysis, and additionally, introduced
a connection to the graph arboricity. ERS proved that the sample complexity of their
algorithm is $\otilde(\eps^{-2}\alpha/d)$, where $\alpha$ is a bound on the graph
arboricity. It is known that $\alpha \leq \sqrt{nd}$, so that yields the $\otilde(\eps^{-2} \sqrt{n/d})$
query complexity bound for general graphs. Unfortunately, the simple algorithm
and analysis are buried deep inside the paper (Section 5.1 of~\cite{ERS19} has the details).
Chapter 10.3.2.2 of Goldreich's book~\cite{G17-book}, Section 3 of Ron's survey~\cite{Ro19-survey},
and Chakrabarty's lecture notes~\cite{Ch21} give nice expositions of the ERS algorithm,
albeit without the arboricity connection and the additional overhead of parameter search.

The aim of this note to give a full presentation of the ERS algorithm with the arboricity
dependence, and explain all the technicalities of the local search. We do not
incur any logarithmic factors and get a simple, streamlined algorithm.

We do not give a comprehensive description of previous work, but it is worth mentioning
a number of results on edge counting and arboricity. Eden-Rosenbaum~\cite{EdRo18} study the problem
of sampling uniform random edges, and Eden-Ron-Rosenbaum~\cite{EdRoRo19} prove the arboricity captures
the complexity of this problems. T{\v{e}}tek-Thorup~\cite{TT22}, Beretta-T{\v{e}}tek~\cite{BT24}, Beretta-Chakrabarty-Seshadhri~\cite{BCS26},
and Chanda~\cite{Ch26} all study the problem of average degree/edge estimation in variants of the standard adjacency list model.
We note that Chanda's algorithms also get the $\alpha/d$ query complexity, albeit with extra $\poly(\eps^{-1}\log n)$ factors.

\subsection{Setup} \label{sec:setup}

We consider the adjacency list query model (refer to Chap. 10 of~\cite{G17-book}).
The graph $G$ is accessed via the following queries.
(1) Vertex query: get a uar $v \in V$, (2) Degree query: for any vertex $v\in V$, returns $v$'s degree $d_v$; and
(3) Neighbor query: for any vertex $v\in V$, return a uar neighbor of $v$. 
The standard model often includes pair queries, but we do not need this query in our algorithms. 

A subtlety pointed out by Beretta-T{\v{e}}tek~\cite{BT24} and subsequently expanded upon by Beretta-Chakrabarty-Seshadhri~\cite{BCS26}
is that the knowledge of $n$ can make a difference. In our setting, we will assume that $n$ is \emph{not} known.

We introduce the concept of graph \emph{arboricity} (Pg. 52 of~\cite{Di-book}).

\begin{definition}[Arboricity] \label{def:arb}
     The arboricity of a graph $G$, denote $\alpha(G)$, is the minimum number of forests required to cover the edge set of $G$. 
 \end{definition}    

The arboricity has played an important role in subgraph counting algorithms, since the seminal work of
Chiba-Nishizeki~\cite{ChNi85}. A closely related quantity called the graph \emph{degeneracy} is 
a $2$-approximation to the graph arboricity. The class of bounded degeneracy graphs is rich and contains
all minor-closed families. Refer to the short survey of Seshadhri~\cite{Se23} for more details.

We will use this fundamental lemma of Chiba-Nishizeki, which we prove for completeness.

\begin{lemma} [Chiba-Nishizeki, Lemma 2 of~\cite{ChNi85}] \label{lem:cn} $\sum_{(u,v) \in E} \min(d_u, d_v) \leq 2m\alpha(G)$.
\end{lemma}
\begin{proof} We partition the edges of $G$ into at most $\alpha(G)$ forests. In each forest, we 
direct edges from children to their parents. This creates a digraph $\vec{G}$. 
For every directed edge $e$ of $\vec{G}$, let $s(e)$ denote
the source vertex of $e$.
Since each vertex of a forest has at most one parent,
each vertex in $\vec{G}$ is the source for at most $\alpha(G)$ edges. 
\begin{align}
\sum_{(u,v) \in E} \min(d_u, d_v) \leq \sum_{e \in E} d_{s(e)} = \sum_{v \in V} \sum_{e: s(e) = v} d_v
= \sum_{v \in V} d_v \sum_{e: s(e) = v} 1 \leq \alpha(G) \sum_{v \in V} d_v = 2m\alpha(G)
\end{align}
\end{proof}

A central construct in subgraph counting algorithms using the arboricity bound is
the \emph{degree orientation}.  We define an ordering on the vertices; tie-breaking
is done using the vertex IDs.

\begin{definition}[The degree ordering] \label{def:prec}
Let $u\prec v$ if $d_u < d_v$    or if $d_u =d_v$ and $id(u)< id(v)$.

Let $G_\prec$ be the DAG obtained by directing edges according to $\prec$, and let $d^+_u$
be the outdegree of $u$ in $G_\prec$.
\end{definition}

This orientation can be used to prove a simple, yet fundamental bound on the graph arboricity.

\begin{lemma} \label{lem:arb} $\alpha(G) \leq \sqrt{2m}$.
\end{lemma}

\begin{proof} Consider the oriented $G_\prec$, as given in \Def{prec}. Let $u$ be 
a vertex of maximum outdegree, and let $\Gamma^+_u$ denote the outneighborhood of $u$.
Observe that 
$\sum_{v \in \Gamma^+_u} d_v \geq \sum_{v \in \Gamma^+_u} d_u \geq \sum_{v \in \Gamma^+_u} d^+_u = (d^+_u)^2$.
Since the sum of degrees is at most $2m$, we conclude that $d^+_u \leq \sqrt{2m}$.

To partition $G$ into at most $\sqrt{2m}$ forests, consider the following labeling of the edges in $G_\prec$.
For every $v$ in $G_{\prec}$, all outgoing edges are labeled with 
a distinct label in $[1,\floor{\sqrt{2m}}]$. For every $i$, the edge set labeled $i$ is a forest (since for every vertex, at most one outgoing edge is chosen
and $G_\prec$ is a DAG).
Thus, the arboricity is at most $\sqrt{2m}$.
\end{proof}

We state the main theorem about estimating the average degree.

\begin{theorem} \label{thm:main} The algorithm \ers{} takes as input adjacency list
access to the graph $G$ (without knowledge of $n$), an upper bound $\alpha$ on the arboricity, and an approximation
parameter $\eps$. With probability $> 2/3$, the query complexity (and running time) is $O(\eps^{-2} \alpha/d)$ 
and the output lies in the range $(1\pm\eps) d$.
\end{theorem}

\section{The ERS algorithm} \label{sec:ers}

\begin{algorithm}
    \caption{{\ers($G,\alpha,\eps$)}}\label{alg:ers}
    \begin{algorithmic}[1]
        \LineComment{Adj. list access to $G$, $n$ is unknown}
        \LineComment{$\alpha$ upper bound to $\alpha(G)$, $\eps < 1/2$, $c$ is large constant}
    \Statex
	\State Initialize $s = c/\eps^2$ and $\tau = \alpha$.
	\While{true}:
		\For{$i$ in $1,\ldots,s$}
        \State Pick uar vertex $u$, and pick uar neighbor $v$ of $u$.
        \State Query their degrees $d_u$ and $d_v$.
        \State If $u \prec v$, set $X_i = 2d_u$.
        \State Else, set $X_i = 0$.
	    \EndFor
    \State If $X = s^{-1}\sum_{i \leq s} X_i > \tau$, output $X$.
    \State Else, reset $s = 2s$ and $\tau = \tau/2$. \label{step:reset}
	\EndWhile
	\end{algorithmic}
\end{algorithm}

The key idea is in the following bounds.

\begin{claim} \label{clm:X} $\EX[X_i] = d$ and $\Var[X_i] \leq 8d\alpha(G)$.
\end{claim}

\begin{proof} The probability of picking vertex $u$ is $1/n$, and the probability
    of picking a neighbor $v$ such that $v \succ u$ is exactly $d^+_u/d_u$. Suppose $u$ is picked.
    If the latter event occurs, $X_i = 2d_u$, otherwise it is zero. With these in hand,
    \begin{align}
        \EX[X_i] = \sum_{u \in V} \frac{1}{n} \cdot \frac{d^+_u}{d_u} \cdot 2d_u = \frac{1}{n} \sum_{u \in V} 2d^+_u = 2m/n = d
    \end{align}
    We now bound $\Var[X_i] = \EX[X^2_i] - (\EX[X_i]^2) \leq \EX[X^2_i]$.     \begin{align}
        \EX[X^2_i] = \sum_{u \in V} \frac{1}{n} \cdot\frac{d^+_u}{d_u} \cdot 4d^2_u = \frac{4}{n} \sum_{u \in V} d^+_u d_u
    \end{align}
    The Chiba-Nishizeki bound of \Lem{cn} enters the picture. Observe that $$\sum_{u \in V} d^+_u d_u = \sum_{u \in V} \sum_{v \succ u} d_u    =\sum_{(u,v) \in E} \min(d_u, d_v)\leq 2m \alpha(G)$$ 
    Hence, $\EX[X^2_i] \leq (4/n) \cdot (2m\alpha(G)) = 8d\alpha(G)$.
\end{proof}

A Markov argument shows that the procedure cannot terminate too early.

\begin{claim} \label{clm:markov} The probability that \ers$(G,\alpha,\eps)$ terminates
    when $\tau > 8d$ is at most $1/4$.
\end{claim}

\begin{proof} By \Clm{X} and linearity of expectation, $\EX[X] = d$. By the Markov inequality,
    for any iteration $\Pr[X > \tau] < d/\tau$. Note that $X > \tau$ is exactly the termination condition. Let $\tau_0<\tau_1< \ldots$ be all the threshold
    values in the running of \ers that are more than $8d$ (in increasing order). Observe that $\tau_0 > 8d$ and $\tau_i = 2^i \tau_0$. By the union bound, the probability that \ers{} terminates when $\tau > 8d$
    \begin{equation}
        \sum_{i \geq 0} \Pr[X > \tau_i] < \sum_{i \geq 0} d/\tau_i = d \sum_{i \geq 0} 1/(2^i \tau_0) \leq (d/8d) \sum_{i \geq 0} 2^{-i} \leq 1/4
    \end{equation}
\end{proof}

Since the procedure will run sufficiently long, it is guaranteed to take enough samples
to converge. We wrap up the proof now.

\begin{proof}[Proof of Theorem~\ref{thm:main}] Note that $s\tau$ is the same in each iteration, and has value exactly $c\alpha/\eps^2$. 
    Consider an iteration where $\tau \leq 8d$. Hence, $s \geq (c/8\eps^2)(\alpha/d)$. By linearity of expectation
    and \Clm{X}, $\EX[X] = d$. By the independence of the $X_i$'s and the variance
    bound of \Clm{X}, $\Var[X] = s^{-2} \sum_{i \leq s} \Var[X_i] \leq 8d\alpha(G)/s \leq 8d\alpha/s$.
    By Chebyshev's inequality, the bound on $s$, and $c$ being large enough,
    \begin{equation}
        \Pr[|X - \EX[X]| > \eps \EX[X]] \leq \Var[X]/(\eps^2 \EX[X]^2) \leq \frac{8d\alpha/s}{\eps^2 d^2} \leq \frac{8\alpha}{\eps^2 d \cdot (c/8\eps^2)(\alpha/d)}
        = 64/c < 1/100
    \end{equation}
    Thus, in any iteration where $\tau \leq 8d$, $X$ is $(1\pm\eps)$-approximation to $d$ with probability at least $99/100$.

    By \Clm{markov}, with probability at least $3/4$, the procedure \ers{} will reach an iteration where
    $\tau \leq 8d$. Condition on this event. By the union bound, the probability that $X \in (1\pm \eps)d$
    for the next four iterations (if they happen), 
    is at least $1-4/100$. After four iterations, $\tau \leq d/2$ and $X \geq (1-\eps)d \geq d/2 \geq \tau$, so
    the procedure will terminate. Taking a union bound over all events, with probability at least $1-1/25-1/4 \geq 2/3$,
    the output will be a $(1\pm\eps)$-approximation to $d$ and \ers{} terminates by the iteration that $\tau \leq d/2$.
    We will condition on this event.

    We now bound the sample complexity. Observe that each sample of $X_i$ requires four queries: one uar vertex query,
    one random neighbor query, and two degree queries. So the total sample (and time) complexity is linear 
    in the total number of $X_i$ samples generated.
    Since $s\tau = (c/\eps^2)\alpha$, the sample complexity of the last iteration is $O(\eps^{-2}\alpha/d)$. The sample
    complexities increase by a factor of 2 in each iteration, so the total sample (and time) complexity is $O(\eps^{-2}\alpha/d)$.
\end{proof}

\subsection{The general case, for unknown arboricity} \label{sec:general}

When the arboricity is unknown (or the case for general graphs), \ers{} as defined will not work.
For this case, we need to assume that $n$ is known. Observe that the input $\alpha$ is only
used to set the threshold $\tau$; naively, one could initialize $\tau$ to $n$.
A direct application of \Thm{main} would lead to a $O(\eps^{-2}n/d)$ complexity.
On the other hand, the ``correct" bound according to \Lem{arb} would be $O(\eps^{-2} \sqrt{m}/d) = O(\eps^{-2} \sqrt{n/d})$. 

A simple tweak to \Step{reset} suffices to get this bound. 
This modified procedure and analysis is also given in Appendix A.2 of Beretta-Chakrabarty-Seshadhri~\cite{BCS26};
we give the procedure and proof for completeness.

\begin{algorithm}
    \caption{{\ersg($G,\eps$)}}\label{alg:ersg}
    \begin{algorithmic}[1]
        \LineComment{Adj. list access to $G$, $n$ is known}
        \LineComment{$\eps < 1/2$, $c$ is large constant}
    \Statex
	\State Initialize $s = c/\eps^2$ and $\tau = n$.
	\While{true}:
		\For{$i$ in $1,\ldots,s$}
        \State Pick uar vertex $u$, and pick uar neighbor $v$ of $u$.
        \State Query their degrees $d_u$ and $d_v$.
        \State If $u \prec v$, set $X_i = 2d_u$.
        \State Else, set $X_i = 0$.
	    \EndFor
    \State If $X = s^{-1}\sum_{i \leq s} X_i > \tau$, output $X$.
        \State Else, reset $s = 2s$ and $\tau = \tau/4$. \Comment{Note the difference from \ers} \label{step:resetg}
	\EndWhile
	\end{algorithmic}
\end{algorithm}

The random variable $X_i$ is the same, and \Clm{X} and \Clm{markov} hold.

\begin{theorem} \label{thm:ersg} The algorithm \ersg{} takes as input adjacency list
access to the graph $G$ (with knowledge of $n$) and an approximation
    parameter $\eps$. With probability $> 2/3$, the query complexity (and running time) is $O(\eps^{-2} \sqrt{n/d})$ 
and the output lies in the range $(1\pm\eps) d$.
\end{theorem}

\begin{proof} The proof is almost identical to that of \Thm{main}. We observe that $s\sqrt{\tau}$ is the same in each iteration, and has value exactly $c\sqrt{n}/\eps^2$. 
    Consider an iteration where $\tau \leq 8d$. Hence, $s \geq (c/(\sqrt{8}\eps^2))\sqrt{n/d}$. As before, $\EX[X] = d$. By the independence of the $X_i$'s and the variance
    bound of \Clm{X}, $\Var[X] = s^{-2} \sum_{i \leq s} \Var[X_i] \leq 8d\alpha(G)/s$. By \Lem{arb}, $\alpha(G) \leq \sqrt{2m}$.
    Hence, $\Var[X] \leq 8\sqrt{2} \sqrt{m} d/s \leq 16\sqrt{m} d/s$.
    We use Chebyshev's inequality and $c$ being large enough. (For convenience, $c'$ denotes another constant.)
    \begin{align}
        \Pr[|X - \EX[X]| > \eps \EX[X]] \leq \Var[X]/(\eps^2 \EX[X]^2) & \leq \frac{16\sqrt{m}d/s}{\eps^2 d^2} \leq \frac{16\sqrt{m}}{\eps^2 d \cdot (c/\sqrt{8}\eps^2)\sqrt{n/d}} \\
        & \leq \frac{\sqrt{m}}{c' \sqrt{nd}} < 1/100
    \end{align}
    Thus, in any iteration where $\tau \leq 8d$, $X$ is $(1\pm\eps)$-approximation to $d$ with probability at least $99/100$.

    By \Clm{markov}, with probability at least $3/4$, the procedure \ers{} will reach an iteration where
    $\tau \leq 8d$. Following the same argument as in the proof of \Thm{main}, 
    with probability at least $1-1/25-1/4 \geq 2/3$,
    the output will be a $(1\pm\eps)$-approximation to $d$ and \ersg{} terminates by the iteration that $\tau \leq d/2$.
    We will condition on this event.

    We now bound the sample complexity. The total sample (and time) complexity is linear 
    in the total number of $X_i$ samples generated.
    Since $s\sqrt{\tau} = (c/\eps^2)\sqrt{n}$ and $\tau \geq d/2$ by the time \ersg{} terminates , the sample complexity of the last iteration is $O(\eps^{-2}\sqrt{n/d})$. The sample
    complexities increase by a factor of 2 in each iteration, so the total sample (and time) complexity is $O(\eps^{-2}\sqrt{n/d})$.
\end{proof}

When $n$ is unknown, \ersg{} cannot choose the starting value of $\tau = n$. We can use
uar vertex samples to estimate $n$ using the birthday paradox (using Theorem 2.1 of Ron-Tsur~\cite{RoTs16}). Overall, this would
lead to a $O(\eps^{-2} \sqrt{n})$ query algorithm. Beretta-Chakrabarty-Seshadhri give an $\Omega(\min(\sqrt{n},n/d))$
lower bound in this setting (Claim 5.3 of~\cite{BCS26}), showing that the knowledge of $n$ is required to get the $\sqrt{n/d}$ complexity.

\bibliographystyle{alpha}
\bibliography{sublinear}

\end{document}